\documentclass[12pt]{article}

\usepackage[russian, english]{babel} 

\usepackage{epsfig}
\usepackage{amsmath, amsthm, amssymb}

\usepackage{hyperref}

\usepackage{pdfsync}  
\usepackage{verbatim}  

    \usepackage{xcolor}

\newcommand{\beq}[1]{\begin{equation} \label{eq:#1}}
\newcommand{\eeq}{\end{equation}}

\newcommand{\bea}[1]{\begin{eqnarray} \label{eq:#1}}
\newcommand{\eea}{\end{eqnarray}}

\newcommand{\N}{{\ensuremath{\mathbb N}}} 
\newcommand{\Z}{{\ensuremath{\mathbb Z}}} 

\newcommand{\R}{{\ensuremath{\mathbb R}}} 

\newcommand{\var}{\delta}  

\newcommand{\DeltaF}[2][]{\mathop{\delta}_{_{#1}}\nolimits\!\left[#2\right]}  
\newcommand{\deltaF}[2][]{\mathop{\delta}_{_{#1}}\nolimits\!\left(#2\right)}  

\newtheorem{thm}{Theorem}    

\begin{document}

\title{\bf Microcanonical quantum cosmology: gauge fixing and functional determinants on a circle}
\author{Dmitry Nesterov and Andrei Barvinsky}
\date{}
\maketitle
\begin{center}
  \hspace{-0mm}{\em Theory Department, Lebedev Physical Institute\\
  Leninsky Prospect 53, Moscow, Russia, 119991}
\\
 \noindent
\end{center}

 \noindent



\newcommand{\Comment}[1]{}  

\newcommand{\Section}[1]{\section {#1} \hspace{\parindent} }
\newcommand{\EndSection}[1]{\newpage}

\newcommand{\SubSection}[1]{\subsection {#1}  \hspace{\parindent} }
\newcommand{\EndSubSection}[1]{}

\newcommand{\mt}[1]{$#1$}
\renewcommand{\math}[1]{\(#1\)}
\newcommand{\Math}[1]{\[\displaystyle{#1}\]}

\newcommand{\FE}{\mathcal{F}} 
\newcommand{\incr}{{\vartriangle}}
\newcommand{\dd}{\boxminus} 

\newcommand{\NC}{\mathcal{C}}  
\newcommand{\Det}{\mathop{\mathrm{Det}}}
\newcommand{\ToDo}{{\color{red}ToDo }}

\begin{abstract}
Gauge fixing procedure in the path integral for the microcanonical statistical sum in quantum cosmology is considered in a special class of gauge conditions free from residual gauge ambiguities. Relevant functional determinants and functional delta-functions on a circle are considered.
\end{abstract}

\section{Introduction}

The goal of this short note is to revisit the one-loop approximation for the microcanonical partition function in quantum cosmology \cite{Barvinsky:2007vb,Barvinsky:2010yx}. Application of the density matrix in spatially closed cosmology shows that its statistical sum is mainly determined by the minisuperspace sector of the theory described by the scale factor $a$ and lapse function $N$ of the Euclidean FRW metric
    \begin{equation}
    ds^2 = N^2(\tau)d\tau^2 + a^2(\tau)d^2\Omega^{(3)}. \label{metric3}
    \end{equation}
Its time-parametrization invariant action has the form
    \begin{eqnarray}
    &&\varGamma[\,a,N\,]=
    \oint d\tau\,N {\cal L}(a,\dot a/N)+ F(\eta),\quad
    \eta=\oint d\tau\,\frac{N}a.                   \label{eta}
    \end{eqnarray}
Here ${\cal L}(a,\dot a/N)$ is a local Lagrangian and $F(\eta)$ is a free energy of quantum matter fields as a function of the effective inverse temperature $\eta$ -- the period of the Euclidean time $\tau$ measured in units of the conformal time parameter. Integration over $\tau$ runs over this period and is denoted by $\oint$. Remarkable property of this expression is that for {\em any} Lagrangian ${\cal L}(a,\dot a/N)$ the quadratic part of $\varGamma[\,a,N\,]$ on the classical background has a generic form \cite{Barvinsky:2010yx}
  \bea{CosmologicalQuadraticActionN1}
  &&\varGamma_{2} [n,\varphi]
  =\varepsilon\frac12 \oint d\tau \; \left( \dot\varphi^2 + \frac{\ddot{g}}{g}\varphi^2 - 2g \dot{\varphi} n + 2 \dot{g} \varphi n + g^2 n^2\right)
  + \frac{1}{2}\, \NC \left(\oint d\tau \, n \right)^2,
  \label{CosmologicalQuadraticActionN1}\\
  &&g\equiv a\dot{a}\sqrt{|\mathcal{D}|},\quad \mathcal{D}\equiv \frac{\partial^2 L(a,\dot{a})}{{\partial \dot{a}}^2},\quad \varepsilon\equiv\frac{\mathcal{D}}{|\mathcal{D}|},\quad
  \NC=\frac{d^2F}{d\eta^2},
 \eea
where the perturbations of the scale factor and lapse are parameterized by new variables\footnote{New variables are related to perturbations $\delta a$ and $\delta N$ as:
$n=\frac{1}{a}\left(\frac{\var N}{N} - \frac{\var a}{a}\right), \quad \varphi =\sqrt{|\mathcal{D}|}\, \var a$.} $\varphi$ and $n$, and $N=1$ gauge is used for background. The coefficients of this form are built in terms of the function $g=g(\tau)$ and the parameter $\NC$ -- functionals of the gravitational background characterized by a periodic scale factor $a(\tau)$. This background forms the cosmological instanton with the topology of $S^1\times S^3$. The $S^1$ part is associated with the Euclidean time $\tau$ which parameterizes the instanton metric (\ref{metric3}) in the gauge $N=1$, and $a(\tau)$ is a periodically oscillating radius of the sphere $S^3$.


The action (\ref{CosmologicalQuadraticActionN1}) is gauge invariant under the local coordinate transformations
     \begin{equation}
     \var n(\tau) = -\, \dot{f}(\tau),
     \qquad   \var \varphi (\tau) =
     - \,g (\tau) \,{f(\tau)}\;. \label{trans}
     \end{equation}
Therefore, the path integral for the statistical sum of the model requires Faddeev-Popov gauge fixing procedure which was performed in \cite{Barvinsky:2010yx} in the ``relativistic" gauge
     $\chi(n,\varphi)\equiv\dot n=0$.

Since this gauge leaves unfixed residual gauge transformations associated with conformal Killing vectors of the background (\ref{metric3}), additional gauge fixing procedure was necessary, which effectively has led to linearly dependent generators treated by the Batalin-Vilkovisky technique of reducible gauge theories \cite{BatalinVilkovisky}. This makes further analysis of the theory (in particular, the convexity issue for the action of the theory at the saddle points of the path integral) very cumbersome. So here we replace the gauge fixing procedure of \cite{Barvinsky:2010yx} in the gauge (\ref{gauge0}) by using the alternative gauge free from residual gauge transformations.

Our main result will be a calculation of the gauge-fixed Gaussian path integral for the one-loop prefactor $P$ of the statistical sum
    \begin{equation} \label{eq:Z}
     Z = \,P\,e^{-\varGamma_{0}},
    \end{equation}
which is semiclassically dominated by the tree-level contribution $e^{\varGamma_{0}}$ at the saddle point at which the on-shell value of the action (\ref{eta}) is given by $\varGamma_{0}$. This involves treatment of a special delta-function type gauge, which is nonlocal in the Euclidean time, and finding closed algorithms for functional determinants of nonlocal operators on spaces of periodic functions. So this calculation turns out to be technically illuminating and deserves presentation.

\section{Gauge fixing}

Since the system is a gauge one the correct expression for one-loop contribution according to Faddeev-Popov technique is
    \begin{equation} \label{eq:P}
    P =  \Det Q[\chi] \int D n \int D \varphi\; \DeltaF{\chi}\; e^{-\varGamma_{2} [n,\varphi]}
    \end{equation}
where $\chi\big(n,\varphi\big)$ is some gauge condition completely fixing the gauge invariance, and $Q$ -- corresponding Faddeev-Popov determinant. The integration is performed over minisuperspace fields $n (\tau),\, \varphi (\tau)$ defined on a circle\footnote{Functional space of (real-valued) functions on a circle $S^1$ of length $T$ is equivalent to that of (real-valued) functions on $R$ periodic with period $T$. so somewhere we call such functions ``periodic''.} $S^1$ of length $T$.

One of possible gauges (a family of gauges) completely fixing the gauge ambiguity is
    \beq{gauge}
    \chi(n,\varphi)=\beta\dot{n}-\gamma\oint k\,\varphi
    \eeq
where $k(\tau)$ is some function on $S^1$ and $\beta,\gamma$ are some constants.\footnote{The result easily generalizes for $\beta,\gamma \to \beta(\tau),\gamma(\tau)$ with $\beta(\tau)\neq0$ and $\oint \frac{\gamma(\tau)}{\beta(\tau)}\neq0$. In particular this generalization allows to make Faddeev-Popov operator symmetric. But on the other hand it causes complications in calculations. In the article we preserve these constants as the trace of this possible generalization and for the reason of additional dimensionality check.}
In what follows for the sake of convenience we will omit denoting the measure $d\tau$ in all further integrals $\oint d\tau$ where it can not lead to ambiguity.

Absence of residual gauge invariance in this gauge follows from a simple fact that for periodic $n(\tau)$ (i.e. $n(\tau)$ on $S^1$) the equation $\chi(n,\varphi)=0$ implies that both terms of (\ref{eq:gauge}) separately vanish,
    \begin{equation}\label{eq:gauge_decouple}
    \chi(n,\varphi)=0\quad\to\quad \dot n=0,\quad \oint\,k\varphi=0,
    \end{equation}
and thus forbid transformations (\ref{trans}) with a periodic $f(\tau)$.

The Faddeev-Popov operator \math{Q} for such gauge $\chi$, acting on the gauge transformation parameter $f$, is defined by the equation
    \beq{Q}
     Qf(\tau) \;=\; -\beta\frac{d^2f(\tau)}{d\tau^2}+\gamma\oint d\tau'\,k(\tau')\,g(\tau')\,f(\tau').
    \eeq
It does not have periodic zero modes and, therefore, non-degenerate. Its determinant on the space of square integrable functions on a circle $L^2(S^1)$
    \beq{detQ}
     |\Det Q | \;=\; \left|\frac{\gamma}{\beta}\,T^2{\textstyle \oint k\,g} \right|.
    \eeq
This can be shown by direct spectrum analysis (see Appendix \ref{sect:Det}).

Calculation of the rest part of the prefactor
    \beq{integraldelta}
     \int D n \int D \varphi\; \DeltaF{\chi}\; e^{-\varGamma_{2} [n,\varphi]}
    \eeq
simplifies due to the decoupling of the of functional delta-function,
    \beq{}
     \DeltaF{\beta\dot{n} + \gamma\textstyle{\oint} k\,\varphi} =\frac{\sqrt{2\pi}}{\sqrt{T}}\;
     \DeltaF[\Comment{\text{inh sector of }n}]{\beta\dot{n}} \;\deltaF[\Comment{\text{some subsector of } \varphi}]{\gamma\textstyle{\oint} k\,\varphi},
    \eeq
which is a corollary of the relation (\ref{eq:gauge_decouple}). Nontrivial element here is the proportionality coefficient in the right hand side and the definition of the rule how the functional delta-function $\DeltaF[\Comment{\text{inh sector of }n}]{\beta\dot{n}}$ acts on a test functional $f[n]$  in the path integral
    \bea{theorem20}
    \int\limits_{L^2(S^1)}\!\!\!\! Dn\;\tilde{\delta}[\beta\dot{n}]\; f[n]
    \; = \;
    \frac{|\beta|}{\sqrt{T}}
    \int\limits_{\R} \frac{dc}{\sqrt{2\pi}} \;f[c].
    \eea
Here the integration over square-integrable functions on a circle (``periodic'' functions) is denoted by $L^2(S^1)$ and numerical integration over the real axis is labeled by $\R$. Below we prove this property.

Integrating these delta-functions in (\ref{eq:integraldelta}) we reduce the functional integral over $n(\tau)$ to a numerical integral over its homogeneous mode and restrict $\varphi(\tau)$ to the subspace of functions transversal to $k(\tau)$ which we denote $L^2_{\perp k}(S^1)$
    \bea{reducedrep}
    \frac{|\beta|}{|\gamma|} \frac{1}{\sqrt{\oint k^2}} \frac{1}{T}
    \!  \int\limits_{L^2_{\perp k}(S^1)}\!\!\!\!\! D{\varphi}_{\perp k}
    \int\limits_{\R} \!\frac{dc}{\sqrt{2\pi}}
    \; \exp\left\{{ \small  -\frac{\varepsilon}{2}  \big( \varphi_{\perp k}(\tau)\,,\, c\big)
    \left( \!\!\begin{array}{cc}
     F\,\delta(\tau {-} \tau') &  \!\!\! 2\dot{g}(\tau)\\
     2\dot{g} (\tau') &  \!\!\!\left( \oint g^2 + \frac{T^2}{\varepsilon}\NC\right)
    \end{array}\!\!\right)
    \left( \!\!\begin{array}{c}
      \!\varphi_{\perp k}(\tau')\,\\
       c
    \end{array}\!\!\!\!\right)\!\!}\;
    \right\}
    \eea
where
   \begin{eqnarray} \label{eq:F}
   F=-\frac{d^2}{d\tau^2}
   +\frac{\ddot{g}(\tau)}{g(\tau)}
   \end{eqnarray}
is the operator introduced in \cite{Barvinsky:2011hv}, and integration over repeated time arguments $\tau$ and $\tau'$ is assumed in the exponential of this expression (in accordance with the Einstein rule for the condensed index notation). Factors in the denominators in front of the integral originate from nontrivial arguments of delta-functions. Calculational details for (\ref{eq:reducedrep}), in particular definitions and decoupling of delta-functions, zeta-renormalization and the origin of prefactors can be found in Appendix \ref{sect:PIwDelta}.

The Gaussian integration in (\ref{eq:reducedrep}) gives
\bea{quad_form_det}
 &&\!\!\frac{|\beta|}{|\gamma|} \frac{1}{\sqrt{\oint k^2}} \frac{1}{T}
 \;\, \Big({\Det}_{\perp k}(\varepsilon F)\Big)^{-1/2}
 \!\!\!\times \left( \varepsilon \Big( \oint g^2  + \frac{T^2}{\varepsilon}\NC\Big)
  - \varepsilon \oint\! d\tau \oint\! d\tau'\; 2\dot{g}(\tau) G_k(\tau,\tau') 2\dot{g}(\tau')\right)^{-1/2}
\eea
where $G_k(\tau,\tau')$ -- Green's function \cite{Barvinsky:2011hv} for operator $F$ (\ref{eq:F}). Since the latter is degenerate there is an ambiguity in choosing the Green's function, so $G_k(\tau,\tau')$ is the one, which is orthogonal to $k(\tau)$:\; $\oint d\tau'\; G_k(\tau,\tau') \,k(\tau')=0$.

The determinant in the first factor of (\ref{eq:quad_form_det}), as shown in Appendix \ref{sect:Base}, can be expressed as
    \begin{eqnarray} \label{eq:det_rel_eF}
      {\Det}_{\perp k} (\varepsilon F)
      = \frac{\left(\oint k g\right)^2}{\oint k^2 \oint g^2} \,{\Det}_{\perp g} (\varepsilon F)
      = \frac{\left(\oint k g\right)^2}{\oint k^2 } \frac{1}{\varepsilon}  \incr.
    \end{eqnarray}
where in the last equality we substituted ${\Det}_{\perp g}\, F  = \incr \oint g^2$. The latter was obtained in  \cite{Barvinsky:2011hv} and expresses the determinant of operator with a zero mode on $S^1$ in terms of its the monodromy coefficient $\incr $. The power of $\varepsilon$ was obtained via the zeta-function renormalization.

Using the explicit representation of $G_k(\tau,\tau')$ obtained in \cite{Barvinsky:2011hv} one can show that
   \begin{eqnarray}
      \oint d\tau \oint d\tau'\; \dot{g}(\tau) G_k(\tau,\tau') \dot{g}(\tau')=  \frac14 \left( \oint g^2 - \frac{T^2}{\incr}\right),
   \end{eqnarray}
so the latter factor in (\ref{eq:quad_form_det}) equals
    \beq{1000}
      \frac{\incr^{1/2}}{\varepsilon^{1/2}T}  \left(1 + \frac{\incr}{\varepsilon}\NC \right)^{-1/2}.
    \eeq

Assembling Eqs.(\ref{eq:Q},\ref{eq:quad_form_det},\ref{eq:det_rel_eF},\ref{eq:1000}) one finally gets
\bea{final}
 &&\Det Q \int\Comment{\limits_{L^2(S^1)}} D \varphi(\tau) \int\Comment{\limits_{L^2(S^1)}} D n(\tau)  \; \DeltaF{\beta\dot{n}+\gamma\textstyle\oint k\,\varphi} \;  e^{-\varGamma_{2} [n,\varphi]}
 \;=\; \left( 1+ \frac{\incr}{\varepsilon}\NC \right)^{-1/2}\;,
\eea
which gives the preexponential factor $P$ (\ref{eq:P}) to the amplitude of the Universe (\ref{eq:Z}). Needless to say the result actually is independent of the gauge chosen,

As a check of consistency note that in the absence of nonlocal ``matter'' contribution ($\NC=0$) the contribution of the homogeneous sector is trivial.

\section{Conclusion}
Remarkably, in the model of the CFT driven cosmology with
    \begin{eqnarray}
    \varepsilon=-1,\,\,\,\,\NC\equiv\frac{d^2F}{d\eta^2}<0,\,\,\,\,
    \frac{g''}g<0,
    \end{eqnarray}
the result (\ref{eq:final}) fully confirms the answer for the statistical sum prefactor initially obtained in \cite{Barvinsky:2010yx} for a particular class of instantons\footnote{As follows from a deeper analysis of the behavior of term $\frac{g''}g$ for the instantons under consideration it is not exactly negative definite, as was conjectured in \cite{Barvinsky:2010yx}. This fact complicates investigation of the convexity of the action (using the specific gauge suggested in \cite{Barvinsky:2010yx}, where $-\frac{g''}g$ becomes the potential). However it does not spoils the coincidence of the final one-loop results.}. It reads
    \begin{eqnarray}
     P={\rm const}\times\left(1-I\frac{d^2F}{d\eta^2}\right)^{-1/2},
    \end{eqnarray}
where $\incr\equiv -I/\varepsilon$ is the monodromy factor of the operator (\ref{eq:F}) and a numerical constant coefficient is independent of the parameters of the background instanton and, therefore, completely irrelevant. Calculations of \cite{Barvinsky:2010yx}, where the functional determinants were obtained by the variational method, did not control their overall coefficients which depend on one such parameter -- the period of the Euclidean time $T$ (which is a nontrivial nonlocal functional of the instanton metric). Nontrivial factors of $T$ contained in the quantities (\ref{eq:detQ}), (\ref{eq:quad_form_det}) and (\ref{eq:1000}), however, completely cancel out in the final answer (\ref{eq:final}) and, thus, do not affect the validity of conclusions of \cite{Barvinsky:2010yx}.

In contrast to the technique of \cite{Barvinsky:2010yx} the present gauge fixing procedure avoids residual gauge symmetries and related complications of the quantization formalism for theories with reducible generators \cite{BatalinVilkovisky}. This opens prospects for the analysis of convexity properties of the Euclidean action of the theory and the selection criterion for physically relevant saddle points of the cosmological statistical sum.

\section{Acknowledgements}

The work of Dmitry Nesterov was partially supported by RFBR grant 14-02-01173 and Andrei Barvinsky is grateful to RFBR grant 14-01-00489.

\appendix

\section{Determinant of the Fadeev-Popov operator}
\label{sect:Det}

Here we prove (\ref{eq:detQ}) that for nondegenerate operator $Q$ (\ref{eq:Q})  on $L^2(S^1)$
    \beq{}
     Q f= -\beta\frac{d^2}{d\tau^2}f + \gamma\oint k\,g\,f \;\nonumber
    \eeq
its zeta-regularized determinant is
    \beq{det_q_calc}
     \Det Q = \frac{\gamma}{\beta}\,T^2\oint k\,g.
    \eeq

For constant $\beta,\gamma$ one can prove that by analyzing the spectrum of the operator. For simplicity consider it is positively definite, which is for example guaranteed by $\beta\gamma\oint kg >0$.

Right eigenmodes\footnote{Analogous procedure can be performed explicitly for the left set of eigenmodes. Though the modes are different the spectrum obviously remains the same.} can be constructed from that of operator $-\frac{d^2}{d\tau^2}$:
    \beq{basisS1}
     e_n(\tau)\equiv\left\{\begin{array}{ll}
                       \sqrt{2/T} \sin {\frac{2\pi \tau}{T} |n|} & n\in -\N \\
                       \sqrt{1/T}\cdot 1 & n=0 \\
                       \sqrt{2/T} \cos {\frac{2\pi \tau}{T} n} & n\in \N
                     \end{array}\right. \;.
    \eeq

The structure of $Q$ implies that right eigenfunctions should be of the form
    \beq{RightQEigenbasis}
    \tilde{e}_n(\tau)\equiv\left\{\begin{array}{ll}
                       e_n(\tau) - \alpha_n e_0 & n\in -\N \\
                       e_0 & k=0 \\
                       e_n(\tau) - \alpha_n e_0 & n\in \N
                      \end{array}\right.
    \eeq
where $\alpha_n=\frac{\gamma\oint k g e_n}{\beta n^2 e_0 - \gamma\oint k g e_n}$. Correspondent eigenvalues are
    \beq{QEigenvalues}
    \tilde{\lambda}_n=\left\{\begin{array}{ll}
                       \beta\frac{(2\pi)^2}{T^2} n^2 & n\in -\N \\
                       \gamma\oint kg & k=0 \\
                       \beta\frac{(2\pi)^2}{T^2} n^2 & n\in \N
                      \end{array}\right.
    \eeq
Thus the product of eigenvalues equals to the product of nonzero eigenvalues of operator $-\beta\frac{d^2}{d\tau^2}$ and homogeneous mode eigenvalue $ \gamma\oint kg$.
Denoting operator $-\frac{d^2}{d\tau^2}$ on the space of functions $L^2_{\perp const}(S^1) \equiv \{h(\tau):\; \oint h(\tau) =0\}\,$ by $\dd$  one can write \cite{Barvinsky:2012an}
    \begin{equation}
          \Det Q= \beta^{\zeta_\dd(0)} \Det \dd \cdot \gamma\oint kg\;.
    \end{equation}
$\zeta_\dd(0)$ and $\Det \dd$ can be expressed via Riemann zeta function $\zeta_R(s)$ as  $\zeta_\dd(0)=2\zeta_R(0)=-1$ and
    \begin{eqnarray}\label{eq:lndetdd} 
    \ln{\Det} \dd
    =2\sum\limits_{n=1}^\infty
    \ln\left(\frac{2\pi n}{T}\right)^2
    =4\ln\left(\frac{2\pi}{T}\right)
    \zeta_R(0)-4\zeta'_R(0)=2\ln T
    \end{eqnarray}
which confirms (\ref{eq:det_q_calc}).

\section{Path integral with delta-functions}
\label{sect:PIwDelta}

In this section we discuss more precisely the calculation of the path integral (\ref{eq:integraldelta})
    \beq{}
     \int D n \int D \varphi\; \DeltaF{\chi}\; e^{-\varGamma_{2} [n,\varphi]} \nonumber
    \eeq

{ \color{black} \small Note the measure and $\delta$-function conventions we use in this paper:
    \bea{}
     && \int\limits_{L^2(S^1)} D \varphi\; e^{-\frac12\oint \varphi^2}=1\;;
     \nonumber\\
     && \int\limits_{L^2(S^1)} D \varphi\; f[\varphi(\tau)] = \int\limits_{L^2(S^1)} \prod\limits_{m\in\Z} \frac{d\varphi_m}{\sqrt{2\pi}}  \;f[\sum\limits_{m\in\Z}\varphi_m e_m(\tau)]\;;
     \nonumber\\
     && f[\psi] = \int\limits_{L^2(S^1)}\!\!\!\! D \varphi\; \DeltaF{\varphi-\psi}\; {f [\varphi]}
     =  \int\limits_{L^2(S^1)}\!\!\!\! D \varphi \!\!\!\! \int\limits_{L^2(S^1)}\!\!\!\! D \rho\; e^{i\oint \rho(\varphi-\psi)} \; {f [\varphi]} \nonumber
   \;,
    \eea
so that $\DeltaF{\varphi-\psi} = \prod\limits_{m\in\Z} \sqrt{2\pi} \deltaF{\varphi_{m}{-}\psi_{m}}$ under decomposition $\varphi(\tau)=\sum\limits_{m\in\Z}\varphi_m e_m(\tau)$ w.r.t. basis (\ref{eq:basisS1}).
\\

For ``inhomogeneous'' subspace ${L^2_{\perp const}(S^1)}\equiv \{\tilde{\varphi}(\tau)\,|\; \oint \tilde{\varphi} =0\}$ and correspondent delta-functions:
    \bea{}
     && \int\limits_{L^2(S^1)}\!\!\!\! D \varphi\; f[\varphi] =  \!\!\!\!\int\limits_{L^2_{\perp const}(S^1)}\!\!\!\!\!\!\!\! D \hat{\varphi} \int\limits_{\R} \frac{d\varphi_0}{\sqrt{2\pi}}\;f[\tilde{\varphi}(\tau)+\varphi_0 e_0(\tau)]\;;
     \nonumber\\
     && f[\tilde{\psi}] = \!\!\!\! \int\limits_{L^2_{\perp const}(S^1)}\!\!\!\! D \tilde{\varphi}\;\; \tilde{\delta}[\tilde{\varphi}-\tilde{\psi}] \;{f [\tilde{\varphi}]}
     \;=  \!\!\!\! \int\limits_{L^2_{\perp const}(S^1)}\!\!\!\!\! D \tilde{\varphi} \!\!\!\!
      \int\limits_{L^2_{\perp const}(S^1)}\!\!\!\!\! D \tilde{\rho}\;\; e^{i\oint \tilde{\rho}(\tilde{\varphi}-\tilde{\psi})} \; {f [\tilde{\varphi}]} \nonumber
   \;.
    \eea
\normalsize \normalcolor
}

The first step is the decoupling of the delta-functions for the gauge (\ref{eq:gauge}):
\begin{thm}
    \bea{theorem1}
    \DeltaF{\beta\dot{n} + \gamma\textstyle{\oint} k\,\varphi}
    = \frac{\sqrt{2\pi}}{\sqrt{T}}\;\tilde{\delta}[\beta\dot{n}]  \;\deltaF{\gamma\textstyle{\oint} k\,\varphi}
    \eea
\end{thm}

\begin{proof}
    \bea{}
     && \int\limits_{L^2(S^1)} D n \int\limits_{L^2(S^1)} D \varphi\; \DeltaF{\beta\dot{n} + \gamma\textstyle{\oint} k\,\varphi}\; {f [n,\varphi]}
     \nonumber\\
     && =\int\limits_{L^2(S^1)} D n \int\limits_{L^2(S^1)} D \varphi  \int\limits_{L^2(S^1)} D \rho\;\; e^{i\oint \rho(\beta\dot{n} + \gamma\oint k\,\varphi)}\; {f [n,\varphi]}
     \nonumber\\
     && =\int\limits_{L^2(S^1)} D n \int\limits_{L^2(S^1)} D \varphi  \int\limits_{L^2(S^1)/const} D \tilde{\rho} \; e^{i\beta\oint \tilde{\rho}\dot{n}}
     \int\limits_{\R} \frac{d\rho_0}{\sqrt{2\pi}}\; e^{i \rho_0\sqrt{T}\gamma\oint k\,\varphi}\; {f [n,\varphi]}\;
     \nonumber\\
     && =\int\limits_{L^2(S^1)} D n \; \tilde{\delta}[\,\beta\dot{n}\,]\int\limits_{L^2(S^1)} D \varphi\;  \frac{\sqrt{2\pi}}{\sqrt{T}}\deltaF{\gamma\textstyle{\oint} k\,\varphi} \;{f [n,\varphi]}\;,
    \eea
where we used decomposition of functions on $S^1$ into homogeneous and inhomogeneous subspaces $\rho(\tau)=\rho_0 e_0(\tau)+\tilde{\rho}(\tau)$  and $e_0(\tau)$ - homogeneous basis function (\ref{eq:basisS1})\if{\footnote{Note the measure and $\delta$-function conventions we use in this paper:
    \bea{}
     && \int\limits_{L^2(S^1)} D \varphi\; e^{-\frac12\oint \varphi^2}=1\;;
     \nonumber\\
     && \int\limits_{L^2(S^1)} D \varphi\; f[\varphi(\tau)] = \int\limits_{L^2(S^1)} \prod\limits_{m\in\Z} \frac{d\varphi_m}{\sqrt{2\pi}}  \;f[\sum\limits_{m\in\Z}\varphi_m e_m(\tau)]\;;
     \nonumber\\
     && f[\psi] = \int\limits_{L^2(S^1)}\!\!\!\! D \varphi\; \DeltaF{\varphi-\psi}\; {f [\varphi]}
     =  \int\limits_{L^2(S^1)}\!\!\!\! D \varphi \!\!\!\! \int\limits_{L^2(S^1)}\!\!\!\! D \rho\; e^{i\oint \rho(\varphi-\psi)} \; {f [\varphi]} \nonumber
   \;,
    \eea
so that $\DeltaF{\varphi-\psi} = \prod\limits_{m\in\Z} \sqrt{2\pi} \deltaF{\varphi_{m}{-}\psi_{m}}$ under spectral decomposition $\varphi(\tau)=\sum\limits_{m\in\Z}\varphi_m e_m(\tau)$ w.r.t. basis (\ref{eq:basisS1}).

For ``inhomogeneous'' subspace ${L^2_{\perp const}(S^1)}\equiv \{\tilde{\varphi}\,|\; \oint \tilde{\varphi} =0\}$ and correspondent delta-functions:
    \bea{}
     && \int\limits_{L^2(S^1)}\!\!\!\! D \varphi\; f[\varphi] =  \!\!\!\!\int\limits_{L^2_{\perp const}(S^1)}\!\!\!\!\!\!\!\! D \hat{\varphi} \int\limits_{\R} \frac{d\varphi_0}{\sqrt{2\pi}}\;f[\tilde{\varphi}(\tau)+\varphi_0 e_0(\tau)]\;;
     \nonumber\\
     && f[\tilde{\psi}] = \!\!\!\! \int\limits_{L^2_{\perp const}(S^1)}\!\!\!\! D \tilde{\varphi}\;\; \tilde{\delta}[\tilde{\varphi}-\tilde{\psi}] \;{f [\tilde{\varphi}]}
     \;=  \!\!\!\! \int\limits_{L^2_{\perp const}(S^1)}\!\!\!\!\! D \tilde{\varphi} \!\!\!\!
      \int\limits_{L^2_{\perp const}(S^1)}\!\!\!\!\! D \tilde{\rho}\;\; e^{i\oint \tilde{\rho}(\tilde{\varphi}-\tilde{\psi})} \; {f [\tilde{\varphi}]} \nonumber
   \;.
    \eea
} }\fi

\end{proof}

As was noted, the decoupling of delta-functions could be predicted by the fact that gauge-fixing condition $\beta\dot{n} + \gamma\textstyle{\oint} k\,\varphi=0$ implies separate vanishing of $\dot{n}$ and $\textstyle{\oint} k\,\varphi$. The main purpose of the theorem is to fix extra factors under given normalization of path integrals and delta-functions.
\\

Going further in calculation of (\ref{eq:P}) with the help of two following theorems one can perform integration with correspondent delta-functions in $n$ and $\varphi$ sectors.

\begin{thm}
 On \math{S^1}
    \bea{theorem2}
    \int\limits_{L^2(S^1)}\!\!\!\! Dn\;\tilde{\delta}[\beta\dot{n}]\; f[n]
    \quad = \quad
    \frac{|\beta|}{\sqrt{T}}
    \int\limits_{\R} \frac{dc}{\sqrt{2\pi}} \;f[c]
    \eea
\end{thm}
\begin{proof} Using spectral decomposition \math{\dot{n}(\tau) = \sum\limits_{m\in\Z} c_m \dot{e}_m = \sum\limits_{m\in\Z} c_m \frac{2\pi}{T} (-m) e_{-m}(\tau)} one comes to
    \bea{}
     &\int Dn\;\tilde{\delta}[\beta\dot{n}]\; f[n]
     &=   \int \prod\limits_{k\in\Z} \frac{d c_k}{\sqrt{2\pi}}\;
      \prod\limits_{m\in\Z\neq0} \sqrt{2\pi}\deltaF{\beta\frac{2\pi(+m)}{T}c_{-m}} f[\sum c_m e_m]\;
      \nonumber\\
     &&=  \prod\limits_{m\in\Z\neq0}\!\!\big|{\frac{T}{2\pi m \beta}}\big| \;
          \times\! \int\limits_\R \frac{d c_0}{\sqrt{2\pi}}\;
           f[c_0\,e_0]
      \nonumber\\
     &&=  \frac{|\beta|}{\sqrt{T}}
    \int\limits_{\R} \frac{dc}{\sqrt{2\pi}} \;f[c]
    \eea
Product factor $\prod\limits_{m\in\Z\neq0} \big|{\frac{T}{2\pi m \beta}}\big|$ is nothing but the inverse of (absolute value of) the operator's $\beta \partial_\tau$ determinant on $L^2_{\perp const}(S^1) \equiv \{h(\tau):\; \oint h(\tau) =0\}\,$ -- the space of functions on $S^1$ orthogonal to constant functions. This determinant can be expressed in terms of Riemann zeta-function \cite{Abr-Stegun}:
    \bea{}
     &&\prod\limits_{m\in\Z\neq0} \!\!\big|{\frac{T}{2\pi m \beta}}\big|
     = e^{\sum\limits_{m\in\Z\neq0} \ln{\big|{\frac{T}{2\pi\beta}}\big|}+ \sum\limits_{m\in\Z\neq0} \ln{\big|{\frac{1}{m}}\big|}}
     = e^{\ln{|{\frac{T}{2\pi\beta}}|\cdot 2\zeta_R(0) + 2\zeta'_R(0)}}
     = e^{-\ln{|{\frac{T}{2\pi\beta}}| - \ln{2\pi}}}
     = \frac{|\beta|}{T}\;.\;\;
    \eea
Extra $\sqrt{T}$ comes from rescaling of integration variable $c_0=c\sqrt{T}$.

\end{proof}
Also the product factor in calculations above could be expressed as zeta-renormalized inverse of square root of determinant of operator $-\beta^2\frac{d^2}{d\tau^2}$ (\ref{eq:lndetdd}) on the same space which leads to the same result.
\\


Next theorem shows the result of integrating  one-dimensional delta-function:

\begin{thm}
    \bea{theorem3}
     &&\sqrt{2\pi}\!\!\int\limits_{L^2(S^1)}\!\!\!\! D\varphi\;\deltaF{\gamma\oint k\varphi}\; f[\varphi]
     \quad = \quad
     \frac{1}{\sqrt{\gamma^2\oint k^2}}
    \!\! \int\limits_{L^2_{\perp k}(S^1)}\!\!\!\! D {\varphi_{\perp k}}\; f[{\varphi_{\perp k}}]
    \eea
\end{thm}

\begin{proof}
Here it is convenient to use method of dealing with delta-functions in path integrals analogous to that used in Appendix in \cite{Barvinsky:2010yx}.
    \bea{}
     &\sqrt{2\pi}\!\!\! \displaystyle{\int\limits_{L^2(S^1)}}\!\!\!\! D\varphi\;\deltaF{\gamma\oint k\varphi}\; f[\varphi]
     &\equiv \;
     f\left[{\textstyle\frac{1}{i}\frac{\var}{\var J}}\right]
     \sqrt{2\pi} \lim_{\epsilon\to 0}\!\! \int\limits_{L^2(S^1)}\!\!\!\! D\varphi\;
          \deltaF{\gamma\oint k\varphi}\;
          e^{i\oint\varphi J} e^{-\epsilon\frac12\oint \varphi^2} \Big|_{J(\tau)=0}
     \nonumber\\
     && = \;
     f\left[{\textstyle\frac{1}{i}\frac{\var}{\var J}}\right]
      \lim_{\epsilon\to 0} \!\! \int\limits_{L^2(S^1)}\!\!\!\! D\varphi
     \int\limits_{\R} \frac{d\rho_0}{\sqrt{2\pi}}\; e^{i \rho_0 \gamma\oint k\,\varphi}\;
         e^{i\oint\varphi J} e^{-\epsilon\frac12\oint \varphi^2} \Big|_{J(\tau)=0}
     \\
     && = \;
     f\left[{\textstyle\frac{1}{i}\frac{\var}{\var J}}\right]
      \;\lim_{\epsilon\to 0}
     \left(\Det \frac{\var^2 S[\varphi(\tau),\rho_0]}{\var(\varphi(\tau),\rho_0)^2}\right)^{-1/2}
      e^{-S[\varphi(\tau),\rho_0]\big|_{on-shell}} \Big|_{J(\tau)=0} \nonumber
    \eea
where $S[\varphi(\tau),\rho_0]=\oint\left( - i \rho_0 \gamma\, k\,\varphi - i\,\varphi\, J +\epsilon\frac12 \varphi^2 \right)$. To make integrals well defined we have introduced the regularization of the initial integrals with regularization parameter $\epsilon$ (which is auxiliary and should not be mixed with $\varepsilon$ parameter in the main part of the article).

Equations of motion are $\epsilon \varphi (\tau)= i \gamma \rho_0 k(\tau) + i J(\tau)$ and $\oint k(\tau)\varphi(\tau)=0$, so on-shell one can express $\varphi(\tau)$ and $\rho_0$ in terms of $J(\tau)$ and $k(\tau)$ which gives $\rho_0=-\frac{\oint kJ}{\gamma\oint k^2}$ and $\varphi (\tau)= \frac{i}{\epsilon} P_{\perp k} J(\tau)$, where $P_{\perp k}$ -- projector orthogonal to $k$:
    \bea{}
     P_{\perp k} J(\tau) \equiv \oint d\tau' \left(\deltaF{\tau{-}\tau'} - \frac{k(\tau)k(\tau')}{\oint k^2}\right) J(\tau')\;.
    \eea

So, finally
    \bea{}
     S[\varphi(\tau),\rho_0]\Big|_{on-shell} = \frac{1}{\epsilon} \frac12 \oint  J P_{\perp k} J \;.
    \eea

Determinant of the hessian
    \bea{}
     \frac{\var^2 S[\varphi(\tau),\rho_0]}{\var(\varphi(\tau),\rho_0)\;\var(\varphi(\tau'),\rho_0)}
     =  \left(
                     \begin{array}{cc}
                       \epsilon\deltaF{\tau-\tau'} & -i\gamma k(\tau) \\
                        -i\gamma k(\tau')   &  0\\
                     \end{array}
                   \right)
     \eea
where the quadratic form acts on vector of $L^2(S^1)\oplus \R$ space, by the determinant property of block matrices gives
    \bea{DetTheor3}
     \Det \frac{\var^2 S[\varphi(\tau),\rho_0]}{\var(\varphi(\tau),\rho_0)\;\var(\varphi(\tau'),\rho_0)}
     =  \Det \epsilon \mathbb{I}_{L^2(S^1)} \times \frac{\gamma^2}{\epsilon} \oint k^2\;.
     \eea
Since $\Det \epsilon \mathbb{I}_{L^2(S^1)} \;= \;\epsilon^{\zeta_{\mathbb I}(0)}$ then
    \bea{}
     &&\sqrt{2\pi} \int\limits_{L^2(S^1)}\!\!\!\! D\varphi\;\deltaF{\gamma\oint k\varphi}\; f[\varphi]
     \; = \;
     f\left[{\textstyle\frac{1}{i}\frac{\var}{\var J}}\right]
     \frac{1}{\sqrt{\gamma^2\oint k^2}}
      \;\lim_{\epsilon\to 0}\;
      \frac1{\sqrt{\epsilon}^{\zeta_{\mathbb I}(0)-1}}
      e^{-\frac{1}{\epsilon} \frac12 \oint  J P_{\perp k} J}\Big|_{J(\tau)=0}
    \eea
 $\zeta_{\mathbb I}(0)$ -- quantity ``counting'' number of modes on $L^2(S^1)$. Its regularized value is expressed in terms of Riemann zeta-function $\epsilon^{\zeta_{\mathbb I}(0)}= \epsilon^{2\zeta_R(0) +1} =1$. But in fact there is no need in this here, since we will utilize all $\zeta_\mathbb{I}(0){-}1$ powers of $\sqrt{\epsilon}$ coming from $-1/2$ power of the determinant (\ref{eq:DetTheor3}) to obtain in the limit $\epsilon\to0$  product of $\zeta_\mathbb{I}(0){-}1$ delta-functions for each mode on $L^2_{\perp k}(S^1)$ space.
\footnote{On $L^2(S^1)$ one can introduce some orthonormal basis $\breve{e}_m(\tau)$, $m\in\Z$ so that $\breve{e}_0(\tau)=k(\tau)/|k|$, then $J_{\tau}\equiv\sum_{m\in \Z'} J_m \breve{e}_m(\tau)$ and  $\oint  J P_{\perp k} J = \sum_{m\in \Z'} J_m^2$ where $m\in \Z'$ runs through all the modes except 0-th mode $\breve{e}_0(\tau)$. ``Number'' of all modes in $L^2(S^1)$ basis is  $\zeta_{\mathbb I}(0)$. So $\zeta_{\mathbb I}(0){-}1$ is the number of modes in basis of $L^2_{\perp k}(S^1)$ -- the space of functions orthogonal to $k(\tau)$. Thus
    \bea{}
     &&
      \;\lim_{\epsilon\to 0}\;
      \frac1{\sqrt{\epsilon}^{\zeta_{\mathbb I}(0)-1}}
      e^{-\frac{1}{\epsilon} \frac12 \oint  J P_{\perp k} J}
     =
      \;\lim_{\epsilon\to 0}\; \prod_{m\in\Z'} \left(
      \frac1{\sqrt{\epsilon}} \cdot
      e^{-\frac{1}{\epsilon} \frac12  J_m J_m}\right)
     =
      \prod_{m\in\Z'} \sqrt{2\pi} \deltaF{J_m}
     =
      {\delta}_{\perp k}[P_{\perp k} J(\tau)]\;. \nonumber
    \eea
%
}.

Thus the limit in r.h.s. precisely gives ${\delta}_{\perp k}\left[P_{\perp k} J(\tau)\right]$.
Exploiting the integral representation of the latter
    \bea{}
     &&  {\delta}_{\perp k}[P_{\perp k} J(\tau)]
     \;=  \int\limits_{L^2_{\perp k}(S^1)}\!\!\!\! D \varphi_{\perp k} \; e^{i\oint \varphi_{\perp k} J} \;.
    \eea
applying $f\left[{\textstyle\frac{1}{i}\frac{\var}{\var J}}\right]  e^{i\oint \varphi_{\perp k} J} = f[\varphi_{\perp k}]\; e^{i\oint \varphi_{\perp k} J}$ and setting $J(\tau)$ to zero one finally proves the statement of the theorem (\ref{eq:theorem3})
.

\end{proof}

Combining together three theorems above one gets (\ref{eq:integraldelta}) in the form (\ref{eq:reducedrep})
    \beq{reducedrep_app}
     \int\limits_{L^2(S^1)}\!\!\! D n \int\limits_{L^2(S^1)}\!\!\! D \varphi\; \DeltaF{\beta\dot{n} + \gamma\textstyle{\oint} k\,\varphi}\; e^{-\varGamma_{2} [n,\varphi]}
    =
     \frac{|\beta|}{|\gamma|} \frac{1}{\sqrt{\oint k^2}}
    \frac{1}{T} \!\!\!
    \int\limits_{L^2_{\perp k}(S^1)}\!\!\!\!\! D{\varphi}_{\perp k} \int\limits_{\R} \frac{dc}{\sqrt{2\pi}}
    \; e^{-\varGamma_{2} [c,\varphi_{\perp k}]} \;.
    \eeq

\section{Changing the gauge in Gaussian integration}
\label{sect:Base}

Here we justify the relation (\ref{eq:det_rel_eF}) using the analogy with the finite-dimensional vector case.

 Let $A$ be degenerate finite dimensional symmetric operator on linear space $V$ with one zero eigenvector $|g\!>$.

Determinant of $A$ is zero, but one is often interested in 'regularized' determinant $\det_*A$ which is defined as the product of its nonzero eigenvalues.

The latter object has the integral representation (to simplify the formulae we will assume $A$ does not have negative eigenvalues, and consider the simplest norm in $V$):
   \begin{eqnarray} \label{det_perp_g A}
     &&\big({\det}_* A\big)^{-1/2} =
     \big({\det}_{\perp g} A\big)^{-1/2}
     \equiv \int\limits_{V_{\perp g}}\frac{d^{n-1}x'\;}{(2\pi)^{(n-1)/2}}\;e^{-\frac12 <x'| A |x'>},
    \end{eqnarray}
where $<\!.,.\!>$ is the Euclidean inner product in $V=\R^n$ and $x'$ runs over $V_{\perp g}$ , which is $(n{-}1)$-dimensional subspace of initial $\R^n$ which is orthogonal to $|g\!>$. In the exponent $|x'\!>$ is assumed to be immersed in $V $ as $|x',0\!>$. In $|x\!>=|x',z\!>$ the one-dimensional parameter $z$ is the coordinate along $|g\!>$. Such integral representation can be interpreted as the \emph{determinant of operator over particular subspace} of initial vector space.

One can write the 'covariant' integral over all $V$ introducing the delta-function
   \begin{eqnarray} \label{det_perp_g A_Rn}
     && \big({\det}_{\perp g} A\big)^{-1/2}
     =\int\limits_{V}\frac{d^{n}x }{(2\pi)^{(n-1)/2}}\;\delta(<x|e_g\!>)e^{-\frac12 <x| A |x>}\;,
    \end{eqnarray}
where $|e_g\!>=\frac{|g\!>}{|g|}$. Note that $\delta(<x|e_g\!>)$ is nothing but $\delta(z)$.

Let $k$ be some vector so that $<\! k,g \!>\neq0$ and define the following object -- the determinant of $A$ over subspace, orthogonal to $|e_k\!>=\frac{|k>}{|k|}$:
    \begin{eqnarray} \label{det_perp_k A}
     &\big({\det}_{\perp k} A\big)^{-1/2} &\equiv  \int\limits_{V_{\perp k}}\frac{d^{n-1}x''}{(2\pi)^{(n-1)/2}}\;e^{-\frac12 <x''| A |x''>}
     =\int\limits_{V}\frac{d^{n}x }{(2\pi)^{(n-1)/2}}\;\delta(<x|e_k\!>)
     \;e^{-\frac12 <x| A |x>}\;.
    \end{eqnarray}

The algebraic relation between such objects is
    \begin{eqnarray} \label{det_rel}
      {\det}_{\perp k} A = \frac{<\!k,g\!>^2}{|k|^2\,|g|^2} \,{\det}_{\perp g} A \;.
    \end{eqnarray}

This can be shown via integral representations as follows:
   \begin{eqnarray}
     &\big({\det}_{\perp g} A\big)^{-1/2}
     &= \int\limits_{V}\frac{d^{n-1}x' dz}{(2\pi)^{(n-1)/2}}\;\delta(z)e^{-\frac12 <x',z| A |x',z>}
     \nonumber\\
     &&=\int\limits_{V}\frac{d^{n-1}x' dz}{(2\pi)^{(n-1)/2}}\;\delta(z\!-\!f(x'))
     \;e^{-\frac12 <x',z| A |x',z>}
     \nonumber\\
     &&
     = \int\limits_{V}\frac{d^{n-1}x' dz}{(2\pi)^{(n-1)/2}}\;\delta\big(<x|e_g\!>+\frac{<x'|e_k\!>}{<\!e_g|e_k\!>}\big)
     \;e^{-\frac12 <x',z| A |x',z>}
     \nonumber\\
     &&= \;|<\!e_g|e_k\!>|\int\limits_{V}\frac{d^{n}x }{(2\pi)^{(n-1)/2}}\;\delta(<x|e_k\!>)
     \;e^{-\frac12 <x| A |x>}\;,
     \nonumber
    \end{eqnarray}
where we used the degenerateness of $A$ in $g$-direction:
$$<\! x'|A|x' \!> \,\equiv\, <\! x',0|A|x',0 \!> \,=\, <\! x',f(x')|A|x',f(x')\!> $$
and \; $$<x|e_k\!>=<x|e_g\!><e_g|e_k\!>+<x|\Pi_\perp|e_k\!>=\big(<\!x|e_g\!>+\frac{<\!x'|e_k\!>}{<\!e_g|e_k\!>}\big)<\!e_g|e_k\!>.$$

In the article we use the generalization of the property above for the infinite dimensional vector space $L^2(S^1)$
    \begin{eqnarray} 
      {\Det}_{\perp k} F = \frac{\left(\oint k g\right)^2}{\oint k^2 \oint g^2} \,{\Det}_{\perp g} F \;,
    \end{eqnarray}
where integral definitions of determinants are
    \begin{eqnarray}
     \big({\Det}_{\perp f} F \big)^{-1/2} = \; \int\limits_{L^2_{\perp f}(S^1)} D\varphi_{\perp f}  \;\exp \left\{-\frac{1}{2} \oint \varphi_{\perp f} (\tau) F \varphi_{\perp f} (\tau)\right\}
    \end{eqnarray}
with correspondent $f(\tau)$, and its orthogonal complement $L^2_{\perp f}(S^1)=\{\varphi(\tau)\in L^2(S^1) \,| \; \oint \varphi f=0\}$.
\vskip 5mm

The reasoning for the path integral case is analogous.


\end{document}